\documentclass[runningheads]{llncs}
\usepackage{tikz}
\usetikzlibrary{arrows.meta}
\usepackage{subcaption}
\usepackage{graphicx, amsfonts, amsmath, amssymb, caption, authblk, url, multirow, makecell, framed, float, tikz, xcolor, mathrsfs, enumerate}
\usepackage{algorithm}
\usepackage{algpseudocode}
\usepackage{comment}

\newcommand{\h}[1]{\underline{\mathbf{#1}}}
\newcommand{\ind}[1]{^{(#1)}}
\newcommand{\ssq}[1]{^{[#1]}}
\renewcommand{\cal}{\mathcal}
\newcommand{\mtt}{\mathtt}
\newcommand{\ttt}{\texttt}
\newcommand{\bb}{\mathbb}
\newcommand*{\mybox}[1]{%
  \framebox{\raisebox{0cm}[0.5\baselineskip][0.05\baselineskip]{%
    \hbox to 0.1cm{\hss#1\hss}}}}

\begin{document}

\title{0-Cyclic Equalizability of Binary Words Characterized by Hamming Weight}

\author{Sarunyu Thongjarast\inst{1}\thanks{Corresponding author}}
\authorrunning{S. Thongjarast}
\institute{Massachusetts Institute of Technology, Cambridge, MA, USA \\
\email{thong125@mit.edu}}
\maketitle

\begin{abstract}
The random cut is one of the most fundamental shuffles in card-based cryptography: it rotates a sequence of face-down cards by a secret amount. Under this shuffle, two sequences of cards are indistinguishable if and only if they are cyclic shifts of each other. This motivates the question of whether, given two sequences of cards, inserting cards at matching positions can make them indistinguishable. A previous study~\cite{thongjarast2025} shows that such an insertion is always possible when any cards may be inserted, as long as the two words are permutations of each other. This paper considers a stronger restriction: if the cards are binary, carrying only $0$ or $1$, can we insert only $0$s to make the sequences indistinguishable? We call two words \emph{$0$-cyclically equalizable} if one can insert $0$s into both sequences at matching positions so that the resulting words are cyclic shifts of each other. Our main result is that two binary words of equal length are $0$-cyclically equalizable if and only if they have equal Hamming weight, that is, the same number of $1$-bits. Since equal Hamming weight is clearly necessary, the content of the paper is to show that it is also sufficient. Our proof is constructive: we encode a pair of binary words as a single word over the four-letter alphabet $\{\mtt A, \mtt B, \mtt X, \mtt O\}$, reduce equalizability to a simpler condition in this encoding, and build the required insertion explicitly.
\end{abstract}

\section{Introduction}
Card-based cryptography studies how to carry out secure computations using only a deck of physical cards, with no computer involved. The field began with den~Boer's \emph{five-card trick}~\cite{denboer}, which lets two players compute the logical \textsc{and} of their private bits. The procedure reveals nothing more than the AND value. It was later placed on a rigorous footing by the computational model of Mizuki and Shizuya~\cite{mizuki-shizuya}. In a typical protocol, players arrange face-down cards and then apply a \emph{shuffle}: a randomized operation that scrambles the cards so that, once they are turned face up, the visible arrangement reveals the intended output but nothing about the individual inputs (unless it can be implied by the intended output).

Among these shuffles, the \emph{random cut} is the simplest and most widely used. It cyclically rotates the sequence by an amount that no participant knows. Two sequences of cards that are cyclic shifts of each other become indistinguishable. This property makes many random cut protocols secure.

This motivates a combinatorial question. Suppose two card sequences encode two situations that a protocol must keep indistinguishable. When can we pad them --- by inserting extra helping cards at the same positions in both --- so that the padded sequences are indistinguishable after a random cut? Modeling a sequence of two-colored cards as a binary word, a random cut as a cyclic shift, and the simultaneous padding of both sequences as a \emph{simultaneous insertion} of identical letters at identical positions, we say that two words are \emph{cyclically equalizable} if some simultaneous insertion makes them cyclic shifts of one another.

Because a simultaneous insertion adds the same letters to both words, it preserves the difference between their letter counts, and cyclic shifts have equal letter counts. Equal letter counts are therefore necessary for equalizability, regardless of which symbols one is allowed to insert. The interesting direction is sufficiency. Shinagawa and Nuida~\cite{cyclic} proved that for binary words this single obstruction is also sufficient: two binary words are cyclically equalizable if and only if they have equal Hamming weight. This was later generalized to arbitrary alphabets, where the characterizing invariant becomes the full vector of letter counts; equivalently, the two words must be permutations of each other~\cite{thongjarast2025}.

In all of this prior work, the inserted helping cards may be of any type. From the standpoint of an actual protocol this is a hidden cost, since each distinct symbol corresponds to a distinct kind of helping card that must be supplied. It is therefore natural to ask how restricted the inserted symbols can be. In this paper we push the restriction to its extreme and permit insertions of only one symbol, which we take to be $0$. We call two binary words \emph{$0$-cyclically equalizable} if they can be made cyclic shifts of each other by simultaneously inserting only $0$s.

Our main result is that this most economical form of padding is already as powerful as the general one.

\begin{quote} \textbf{Main Theorem} (see Theorem~\ref{main_theorem}).\itshape{} Two binary words of equal length are $0$-cyclically equalizable if and only if they have the same Hamming weight. \end{quote}

That is, a single type of helping card always suffices; one never needs the other symbol. Necessity is immediate, so the work lies in showing that equal Hamming weight lets us realize the equalization using $0$s alone, and our argument does so constructively, producing an explicit insertion.\footnote{A Python implementation of the construction is available at \url{https://github.com/SThongjarast/Binary-Equalization}.}

The proof encodes the two input words into one. To each position we assign a letter of the alphabet $\{\mtt A, \mtt B, \mtt X, \mtt O\}$ recording the pair of bits carried there: \ttt A and \ttt B mark the two ways the words can disagree (\ttt A when the bit of the first word is $0$ and the bit of the second word is $1$, and \ttt B otherwise), $\mtt X$ marks a shared $1$, and $\mtt O$ marks a shared $0$. Under this encoding, equal Hamming weight becomes the condition $|s|_{\mtt A} = |s|_{\mtt B}$; the requirement that one word be a cyclic shift of the other becomes a local \emph{shift-compatibility} condition; and inserting a $0$ into both words becomes inserting the single letter $\mtt O$. This reduction lets us argue entirely about one four-letter sequence.

We then build the insertion in stages of increasing generality. We first treat sequences over $\{\mtt A, \mtt B\}$. After rotating the sequence to begin at a \emph{valley}, an index from which every cyclic prefix has at least as many $\mtt A$s as $\mtt B$s, we pair each $\mtt A$ with a $\mtt B$ using an interleaving family of arcs, and then insert $\mtt O$s so that every arc spans the same length $d$; once the arcs are equal, the sequence is shift-compatible with offset $d$. We extend this to sequences containing $\mtt X$ by expanding each $\mtt X$ into the pair $\mtt{BA}$, running the $\{\mtt A, \mtt B\}$ construction, and contracting the pairs back. Finally, we handle sequences containing $\mtt O$ by removing the existing $\mtt O$s, equalizing the remainder, and reinserting the $\mtt O$s into the gaps the construction provides.
\bigskip \newline
\textbf{Related Work} Cyclic equalizability was introduced by Shinagawa and Nuida~\cite{cyclic}, who studied it for binary words. The extension to arbitrary alphabets was given in~\cite{thongjarast2025}. Both works focus on the \emph{Parikh vector}, the tuple of letter multiplicities of a word. Words with the same Parikh vector are called \emph{abelian equivalent}, and their study forms the area of abelian combinatorics on words~\cite{ficiPuzynina2022}. Our problem is also tied to the classical notion of \emph{conjugacy} (cyclic equivalence) of words, studied extensively in combinatorics on words~\cite{lothaire}.

\section{Preliminaries}
In this section, we recall the basic notions used throughout the paper. Our terminology follows~\cite{cyclic,thongjarast2025}, with the $\sigma$-restricted variants of simultaneous insertion and cyclic equalizability adapted to our setting.

Let $n$ be a positive integer. We denote by $\bb Z_n$ the set $\{0, 1, \dots, n-1\}$, and we identify $i \mod n$ with the unique element $j \in \bb Z_n$ such that $i \equiv j \pmod n$.

Let $\Sigma$ be a finite alphabet. A \emph{word} over $\Sigma$ is a finite sequence of letters in $\Sigma$. We denote by $\Sigma^*$ the set of all words over $\Sigma$, and by $\varepsilon$ the empty word. For a word $w \in \Sigma^*$, we denote its length by $|w|$ and its $i$-th letter by $w\ind{i}$, using zero-indexing. That is, if $|w| = n$, then $w = w\ind{0}\, w\ind{1} \cdots w\ind{n-1}$.

\begin{definition}[Number of Occurrences]
    Let $\Sigma$ be a finite alphabet, let $c \in \Sigma$, and let 
    $w \in \Sigma^*$. We define $|w|_c$ to be the number of occurrences of $c$ in $w$.
\end{definition}

\begin{definition}[Hamming Weight]
    Let $w$ be a binary word. The Hamming weight of $w$ is $|w|_1$.
\end{definition}

\begin{definition}[Cyclic Equivalence]
    Let $w_1,w_2$ be two words of equal length $n$. We say that $w_1$ and $w_2$ are cyclically equivalent if there exists an integer $d$ such that $w_1\ind i = w_2\ind{i+d \bmod n}$ for all $i = 0,1,\dots,n-1$. In that case, we say that $w_2$ is a cyclic shift of $w_1$ with offset $d$.
\end{definition}

\begin{definition}[Simultaneous Insertion]
    Let $k \ge 2$ and let $w_1, \dots, w_k \in \Sigma^n$ be words of equal length $n$ and let $\sigma$ be a subset of $\Sigma$. A \emph{$\sigma$-simultaneous insertion} is an operation that transforms each word $w_i = a_{i,0}a_{i,1}\cdots a_{i,n-1}$ into the word
\[
u_0 a_{i,0} u_1 a_{i,1} \cdots u_{n-1} a_{i,n-1} u_n,
\]
where each $u_j$ is a word in $\sigma^*$. That is, every word receives 
identical insertions at identical positions.

When $\sigma = \Sigma$, we simply call this a \emph{simultaneous insertion}.
\end{definition}

\begin{definition}[Cyclic Equalizability]
    Let $k \ge 2$ and let $w_1, \dots, w_k \in \Sigma^*$ be words of equal length and let $\sigma \subseteq \Sigma$. We say these words are \emph{$\sigma$-cyclically equalizable} if there exist words $w'_1, \dots, w'_k \in \Sigma^*$ such that
\begin{enumerate}
    \item Each $w'_i$ is obtained from $w_i$ by a $\sigma$-simultaneous insertion, and
    \item The words $w'_1, \dots, w'_k$ are cyclically equivalent.
\end{enumerate}

When $\sigma = \Sigma$, we simply say these words are \emph{cyclically equalizable}.

\end{definition}
In the following examples, bold and underlined letters indicate inserted letters.

For example, the words $1234$ and $2143$ are cyclically equalizable, since we can simultaneously insert letters to obtain
\[
1\underline{\mathbf 4}23\underline{\mathbf 2}4 \quad \text{and} \quad 2\underline{\mathbf 4}14\underline{\mathbf 2}3.
\]
The transformed words are cyclically equivalent with offset 2.

Similarly, the words $0110$ and $1010$ are $0$-cyclically equalizable, since we can simultaneously insert $0$ to obtain
\[
01\underline{\mathbf 0}10 \quad \text{and} \quad 10\underline{\mathbf 0}10.
\]
The transformed words are cyclically equivalent with offset $2$.

In contrast, one can verify that the words $012$ and $021$ are not $1$-cyclically equalizable: we cannot insert $1$ into both words at the same position to make them cyclically equivalent. However, the words are $0$-cyclically equalizable since we can simultaneously insert $0$ to obtain
\[
01\underline{\mathbf 0}2 \quad \text{and} \quad 02\underline{\mathbf 0}1.
\]
The transformed words are cyclically equivalent with offset $2$.

\section{Known Properties of Cyclic Equalizability}
In this section, we review known results on cyclic equalizability that   motivate our work.

\begin{theorem}[\cite{cyclic}]
    Let $u,v \in \{0,1\}^n$ be binary words of equal length. Then, $u$ and $v$ are cyclically equalizable if and only if $|u|_1 = |v|_1$.
\end{theorem}
\begin{theorem}[\cite{thongjarast2025}]
    Let $\Sigma$ be a finite alphabet, and let $u,v \in \Sigma^n$ be words of equal length. Then, $u$ and $v$ are cyclically equalizable if and only if $|u|_c = |v|_c$ for all $c \in \Sigma$.
\end{theorem}

\section{Reformulation of 0-Cyclic Equalizability of Two Binary Words}
In this section, we state our main result and provide a reformulation of the problem that will be used throughout.

\begin{theorem}[Main Theorem]\label{main_theorem}
    Let $w_1,w_2 \in \{0,1\}^n$ be binary words. Then, $w_1$ and $w_2$ are $0$-cyclically equalizable if and only if they have the same Hamming weight.
\end{theorem}

The necessity direction is straightforward. Suppose that $w_1$ and $w_2$ are $0$-cyclically equalizable. Then we can simultaneously insert $0$ to obtain $w_1'$ and $w_2'$ so that $w_1'$ and $w_2'$ are cyclically equivalent. Because inserting $0$s does not change the Hamming weight, we have $|w_1|_1=  |w_1'|_1$ and $|w_2|_1 = |w_2'|_1$. Since $w_1'$ and $w_2'$ are cyclically equivalent, they have equal Hamming weight, and hence $|w_1|_1 = |w_2|_1$.

\subsection{Reformulating the Problem}
Our goal now is to show the sufficient condition: if $w_1,w_2$ are two binary words of equal length and Hamming weight, then they are $0$-cyclically equalizable. We will begin our proof by encoding the two words into one.
\begin{definition}[Encoding]
    Let $w_1,w_2 \in \{0,1\}^n$ be binary words. We define the 
    \emph{encoding} of $(w_1, w_2)$ as the word $s \in \{\mtt A,\mtt B,\mtt X,\mtt O\}^n$ of length $n$ where \[
    s\ind i = \begin{cases}
        \mtt A,\quad &\text{if}\ w_1\ind i = 0\ \text{and}\ w_2\ind i =1\\
        \mtt B, &\text{if}\ w_1\ind i = 1\ \text{and}\ w_2\ind i =0\\
        \mtt X, &\text{if}\ w_1\ind i = 1\ \text{and}\ w_2\ind i =1\\
        \mtt O, &\text{if}\ w_1\ind i = 0\ \text{and}\ w_2\ind i =0.
    \end{cases}
    \]
\end{definition}
\begin{definition}[\ttt{ABXO}-sequence]
    Let $s$ be a word over $\{\mtt{A}, \mtt{B}, \mtt{X}, \mtt{O}\}$. 
    We say $s$ is an \emph{\ttt{ABXO}-sequence} if $|s|_{\mtt{A}} = |s|_{\mtt{B}}$. 
    Similarly, we define \emph{\ttt{AB}-sequences}, \emph{\ttt{ABX}-sequences}, 
    and \emph{\ttt{ABO}-sequences} as words over 
    $\{\mtt{A}, \mtt{B}\}$, $\{\mtt{A}, \mtt{B}, \mtt{X}\}$, 
    and $\{\mtt{A}, \mtt{B}, \mtt{O}\}$, respectively, 
    satisfying the same condition $|s|_{\mtt{A}} = |s|_{\mtt{B}}$.
\end{definition}
\begin{proposition}\label{equalAB}
    Let $w_1, w_2 \in \{0,1\}^n$ be two binary words, and $s$ be their encoding. Then, $w_1$ and $w_2$ have the same Hamming weight if and only if $|s|_{\mtt A} = |s|_{\mtt B}$.
\end{proposition}
\begin{proof}
    Notice that \begin{align*}
    |w_1|_1 &= |s|_{\mtt B} + |s|_{\mtt X}\\
    |w_2|_1 &= |s|_{\mtt A} + |s|_{\mtt X}.
    \end{align*}
    Therefore, the Hamming weights of $w_1$ and $w_2$ are equal if and only if $|s|_{\mtt B} + |s|_{\mtt X} = |s|_{\mtt A} + |s|_{\mtt X}$, which can be simplified as $|s|_{\mtt A} = |s|_{\mtt B}$. \qed
\end{proof}
\begin{example}
Let $n = 6$ and take the words $w_1 = 011010$ and $w_2 = 110010.$ This gives $s = \mtt{AXBOXO}$. Here $|s|_{\mtt A} = |s|_{\mtt B} = 1$, matching the
equal Hamming weights $|w_1|_1 = |w_2|_1 = 3$.
\end{example}
\begin{definition}[Shift-Compatible and Shift-Equalizable]
    Let $s$ be an \ttt{ABXO}-sequence of length $n$ and let $d$ be an integer. We call $s$ \emph{shift-compatible} with offset $d$ if for all $i = 0,1,\dots, n-1$, we have
    \begin{itemize}
        \item if $s\ind{i} \in \{\mtt{A}, \mtt{X}\}$, then $s\ind{i+d \bmod n} \in \{\mtt{X}, \mtt{B}\}$, and
        \item if $s\ind{i} \in \{\mtt{B}, \mtt{O}\}$, then $s\ind{i+d \bmod n} \in \{\mtt{A}, \mtt{O}\}$.
    \end{itemize}
    Additionally, if we can obtain $s'$ by inserting \ttt O into an \ttt{ABXO}-sequence $s$ so that $s'$ is shift-compatible, we call $s$ \emph{\ttt O-shift-equalizable}.
\end{definition}
At this point, the shift-compatibility condition looks complicated. In fact, it is simpler than it looks: Theorem~\ref{sufficient shift compatibility condition} shows that the second condition follows automatically from the first.
\begin{proposition}\label{encoding-shift-equivalence}
    Let $w_1, w_2 \in \{0,1\}^n$ be two binary words, let $s$ be their encoding, and let $d$ be an integer. Then, $w_2$ is a cyclic shift of $w_1$ with offset $d$ if and only if $s$ is shift-compatible with offset $-d$.
\end{proposition}
\begin{proof}
    Suppose that $w_2$ is a cyclic shift of $w_1$ with offset $d$. For each $i$, if $s\ind i \in \{\mtt A, \mtt X\}$, then $w_2\ind i = 1$. Thus, $w_1\ind {i-d \bmod n} = 1$, implying that $s \ind {i-d \bmod n} \in \{\mtt X, \mtt B\}$. Similarly, if $s\ind i \in \{\mtt B, \mtt O\}$, then $w_2\ind i = 0$. Thus, $w_1\ind {i-d \bmod n} = 0$, implying that $s \ind {i-d \bmod n} \in \{\mtt A, \mtt O\}$. By definition, $s$ is shift-compatible with offset $-d$. The forward direction is proven. 

    Suppose that $s$ is shift-compatible with offset $-d$. For each $i$, if $w_2\ind i = 0$, then $s\ind i \in \{\mtt B, \mtt O\}$, so $s \ind{i-d\bmod n} \in \{\mtt A, \mtt O\}$ and $w_1\ind{i-d \bmod n} = 0$. Similarly, if $w_2 \ind i = 1$, then $s\ind i \in \{\mtt A, \mtt X\}$, so $s \ind {i-d \bmod n} \in \{\mtt X, \mtt B\}$ and $w_1\ind{i-d \bmod n} = 1$. This implies that $w_1$ is a cyclic shift of $w_2$ with offset $-d$, which is equivalent to saying that $w_2$ is a cyclic shift of $w_1$ with offset $d$. The backward direction is proven. \qed
\end{proof}

\begin{proposition}\label{problemreformulation}
    Let $w_1,w_2 \in \{0,1\}^n$ be two binary words of equal Hamming weight and let $s$ be their encoding. Then, $w_1$ and $w_2$ are $0$-cyclically equalizable if and only if $s$ is \ttt O-shift-equalizable.
\end{proposition}
\begin{proof}
    Inserting $0$ simultaneously into $w_1$ and $w_2$ at the same 
    position corresponds exactly to inserting $\mtt{O}$ at that 
    position in $s$, since $0$ encodes as $\mtt{O}$ in both words. 
    The result then follows from Proposition \ref{encoding-shift-equivalence}. \qed
\end{proof}
\begin{example}
Let $w_1 = 011010$ and $w_2 = 110100$. Note that $w_2$ is a cyclic shift of
$w_1$ with offset $5$. The encoding gives $s = \mtt{AXBABO}$. We check that $s$ is shift-compatible with offset $1$:
\[
\begin{array}{c|cccccc}
i & 0 & 1 & 2 & 3 & 4 & 5\\
\hline
s\ind{i} & \mtt A & \mtt X & \mtt B & \mtt A & \mtt B & \mtt O\\
s\ind{i+1 \bmod 6} & \mtt X & \mtt B & \mtt A & \mtt B & \mtt O & \mtt A
\end{array}
\]
Each column obeys the rule:
$\mtt A \to \mtt X$, $\mtt X \to \mtt B$, $\mtt B \to \mtt A$,
$\mtt A \to \mtt B$, $\mtt B \to \mtt O$, and $\mtt O \to \mtt A$.
\end{example}
With Propositions~\ref{equalAB} and~\ref{problemreformulation}, 
we can shift our focus from binary words to $\mathtt{ABXO}$-sequences, as shown in the following theorem, which is equivalent to the sufficiency direction of Theorem~\ref{main_theorem}.
\begin{theorem}[Reformulated Main Theorem]
    Let $s$ be an \ttt{ABXO}-sequence. Then, $s$ is \ttt O-shift-equalizable.
\end{theorem}

\section{Construction of Insertion for \ttt{AB}-Sequences}
In this section, we prove that every \ttt{AB}-sequence is 
\ttt{O}-shift-equalizable. That is, we handle the case where the initial sequence contains only \ttt As and \ttt Bs.

\subsection{Intuition and Construction Overview}
Let $s$ be an \ttt{AB}-sequence. Suppose we obtain $s'$ (an \ttt{ABO}-sequence) of length $n'$ by inserting \ttt O into $s$ so that $s'$ is shift-compatible with offset $d$. Because there are no \ttt Xs in $s'$, if ${s'}\ind i = \mtt A$, it must follow that ${s'}\ind {i+d \bmod n'} = \mtt B$. This establishes a bijection between \ttt As and \ttt Bs in $s'$, and hence in $s$.

Therefore, our proof has two main steps:
\begin{enumerate}
    \item Construct a bijection between \ttt As and \ttt Bs in $s$.
    \item Insert $\mtt{O}$s into $s$ so that each corresponding $\mtt{A}$-$\mtt{B}$ pair in $s'$ is separated by the same offset $d$.
\end{enumerate}

\begin{figure}[h!]
\centering
\begin{subfigure}{\textwidth}
\centering
\begin{tikzpicture}[
    letter/.style={draw, circle, minimum size=8mm, inner sep=0pt,
                   font=\ttfamily},
    idx/.style={font=\scriptsize, text=gray},
    arc/.style={-{Stealth[length=2mm]}, thick}
  ]
  \foreach \i/\c in {0/A,1/O,2/A,3/B,4/A,5/B,6/O,7/B}
    \node[letter] (n\i) at (\i*0.95,0) {\c};
  \foreach \i in {0,...,7}
    \node[idx] at (\i*0.95,-0.5) {\i};
  \draw[arc] (n0) to[bend left=50] (n3);
  \draw[arc] (n2) to[bend left=50] (n5);
  \draw[arc] (n4) to[bend left=50] (n7);
\end{tikzpicture}
\caption{After inserting two $\mtt O$s we obtain $s' = \mtt{AOABABOB}$. Every
$\mtt A$ maps to a $\mtt B$ exactly $3$ positions to its right, so $s'$ is
shift-compatible with offset $d = 3$.}
\label{fig:arcs-inserted}
\end{subfigure}

\vspace{1.5em}

\begin{subfigure}{\textwidth}
\centering
\begin{tikzpicture}[
    letter/.style={draw, circle, minimum size=8mm, inner sep=0pt,
                   font=\ttfamily},
    idx/.style={font=\scriptsize, text=gray},
    arc/.style={-{Stealth[length=2mm]}, thick}
  ]
  \foreach \i/\c in {0/A,1/A,2/B,3/A,4/B,5/B}
    \node[letter] (m\i) at (\i*0.95,0) {\c};
  \foreach \i in {0,...,5}
    \node[idx] at (\i*0.95,-0.5) {\i};
  \draw[arc] (m0) to[bend left=50] (m2);
  \draw[arc] (m1) to[bend left=50] (m4);
  \draw[arc] (m3) to[bend left=50] (m5);
\end{tikzpicture}
\caption{The original sequence $s = \mtt{AABABB}$ carries the same
$\mtt A$-$\mtt B$ bijection. The arcs are now of unequal length; inserting
$\mtt O$s stretches each one to the common offset.}
\label{fig:arcs-original}
\end{subfigure}

\caption{Equalizing $s = \mtt{AABABB}$. The bijection between $\mtt A$s and
$\mtt B$s is fixed first~(b); inserting $\mtt O$s then stretches every arc to
the common offset $d = 3$~(a).}
\label{fig:arcs}
\end{figure}
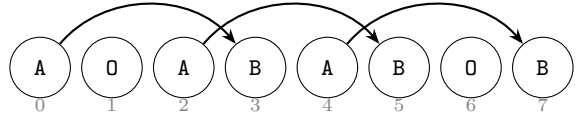
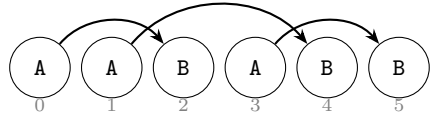
\subsection{Cyclic Shift of the Initial Sequence}
Note that if $s_1$ and $s_2$ are cyclic shifts of each other, then $s_1$ is $\mtt{O}$-shift-equalizable if and only if $s_2$ is. Therefore, without loss of generality, we may replace $s$ by any of its cyclic shifts. In this section, we identify a cyclic shift of $s$ that will be convenient for our construction.

\begin{definition}[Valley]
    Let $s$ be an \ttt{AB}-sequence of length $n$. An index $a \in \mathbb{Z}_n$ is called a \emph{valley} if for every $k = 0,1, 2, \dots, n-1$, the substring \[
        s\ind{a}s\ind{a+1 \bmod n}\dots s\ind{a+k-1 \bmod n}
    \]
    contains at least as many $\mtt{A}$s as $\mtt{B}$s.
\end{definition}
\begin{lemma}\label{Existence of Valley}
    Every \ttt{AB}-sequence has at least one valley.
\end{lemma}
\begin{proof}
    Let $s$ be an \ttt{AB}-sequence of length $n$. Define $s\ssq{a:b}$ as \[s\ind{a}s\ind{a+1 \bmod n}\dots s\ind{b-1 \bmod n}.\] Define $f(i) = |s\ssq{0:i}|_{\mtt{A}} - |s\ssq{0:i}|_{\mtt{B}}$. Let $a$ be an index 
    where $f$ achieves its minimum. We claim that $a$ is a valley. That is, $|s\ssq{a:a+k}|_{\mtt A} - |s\ssq{a:a+k}|_{\mtt B} \geq 0$ for all $k \in \bb Z_n$.

    If $a + k < n$, then \[|s\ssq{a:a+k}|_{\mtt A} - |s\ssq{a:a+k}|_{\mtt B} = f(a+k) - f(a),\] which is non-negative because $f(a)  =\min_{a' \in \bb Z_n} f(a')$.
    
    If $a + k \ge n$, then since $f(n) = 0$ (as $s$ is an \ttt{AB}-sequence),
    \[
    |s\ssq{a:a+k}|_{\mtt{A}} - |s\ssq{a:a+k}|_{\mtt{B}} 
    = (f(n) - f(a)) + f(a+k-n) = f(a+k-n) - f(a),
    \]
    which is non-negative for the same reason.

    Therefore, $a$ is indeed a valley. \qed
\end{proof}
\begin{example}
Consider the \ttt{AB}-sequence $s = \mtt{ABABBABA}$ of length $8$. Reading each
$\mtt A$ as a step of $+1$ and each $\mtt B$ as a step of $-1$, the prefix-sum
function $f(i) = |s\ssq{0:i}|_{\mtt A} - |s\ssq{0:i}|_{\mtt B}$ takes the values
\[
\begin{array}{c|ccccccccc}
i    & 0 & 1 & 2 & 3 & 4 & 5 & 6 & 7 & 8\\
\hline
f(i) & 0 & 1 & 0 & 1 & 0 & -1 & 0 & -1 & 0
\end{array}
\]
shown below as a path that rises on each $\mtt A$ and falls on each $\mtt B$. The
minimum value $-1$ is attained at $i = 5$ and $i = 7$. Therefore, both indices $5$ and $7$ are valleys.

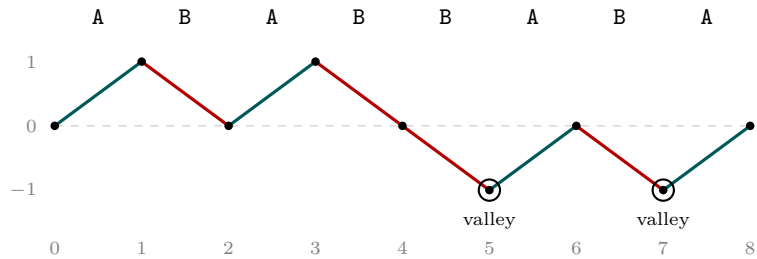
\begin{center}
\begin{tikzpicture}[x=1.15cm, y=0.85cm]
  \draw[gray!45, dashed] (0,0) -- (8,0);
  \foreach \y in {-1,0,1}
    \node[font=\scriptsize, text=gray, left] at (-0.1,\y) {$\y$};
  \foreach \i in {0,...,8}
    \node[font=\scriptsize, text=gray] at (\i,-1.9) {\i};
  \draw[teal!70!black, very thick] (0,0)--(1,1);
  \draw[red!70!black,  very thick] (1,1)--(2,0);
  \draw[teal!70!black, very thick] (2,0)--(3,1);
  \draw[red!70!black,  very thick] (3,1)--(4,0);
  \draw[red!70!black,  very thick] (4,0)--(5,-1);
  \draw[teal!70!black, very thick] (5,-1)--(6,0);
  \draw[red!70!black,  very thick] (6,0)--(7,-1);
  \draw[teal!70!black, very thick] (7,-1)--(8,0);
  \foreach \i/\h in {0/0,1/1,2/0,3/1,4/0,5/-1,6/0,7/-1,8/0}
    \fill (\i,\h) circle (1.6pt);
  \foreach \i/\c in {0/A,1/B,2/A,3/B,4/B,5/A,6/B,7/A}
    \node[font=\ttfamily\small] at (\i+0.5,1.7) {\c};
  \draw[thick] (5,-1) circle (4pt);
  \draw[thick] (7,-1) circle (4pt);
  \node[font=\scriptsize, below] at (5,-1.2) {valley};
  \node[font=\scriptsize, below] at (7,-1.2) {valley};
\end{tikzpicture}
\end{center}
\end{example}

\subsection{Bijection Construction}
In this section, we construct a bijection between $\mtt{A}$s and $\mtt{B}$s in $s$. To motivate the construction, consider arranging the letters of $s'$ on a circle. Since $s'$ is shift-compatible, each $\mtt{A}$ is paired with a $\mtt{B}$ at a fixed offset $d$, and we can draw a chord connecting each such pair. One can observe that these chords cross each other in an interleaving manner. The same crossing structure must therefore appear in $s$ before any $\mtt{O}$s are inserted. 

First, we formalize what a chord is in the following definition.

\begin{definition}[Arcs]
    Let $n$ be a positive integer and let $i, j \in \mathbb{Z}_n$. The \textbf{arc} from $i$ to $j$, denoted $r_n(i, j)$, is the set of indices
    \[
    r_n(i,j) = \{i \bmod n,\ i+1 \bmod n,\ \dots,\ j-1 \bmod n\},
    \] 
    and $i$ and $j$ are called the \emph{start} and \emph{end} of that arc, respectively. 
    
    We say that arc $r_n(i_1, j_1)$ is \emph{contained in} arc $r_n(i_2, j_2)$ if $r_n(i_1, j_1) \subseteq r_n(i_2, j_2)$.
\end{definition}

Next, we define how chords interleave each other. 

\begin{definition}[Interleaving Arcs]
    Let $s$ be an \ttt{ABO}-sequence of length $n$ with $|s|_{\mtt A} + |s|_{\mtt B} = n'$. We say that a sequence of arcs $\cal R = (r_n(i_k,j_k))_{k=0}^{n'/2-1}$ is \emph{interleaving} if it satisfies the following conditions:
    \begin{enumerate}[(i)]
    \item For each $k = 0, 1, \dots, n'/2-1$, we have $i_k < j_k$. That is, the arc does not wrap around the word.
    \item $i_0 < i_1 < \cdots < i_{n'/2-1}$ and $j_0 < j_1 < \cdots < j_{n'/2-1}$.
    \item \label{starting_ending} Each letter \ttt A is the start of an arc in $\cal R$, and each letter \ttt B is the end of an arc in $\cal R$.
\end{enumerate}
\end{definition}

Notice that a sequence of interleaving arcs $\cal R = (r_n(i_k,j_k))_{k=0}^{n'/2-1}$ establishes a bijection between \ttt As and \ttt Bs in $s$. More specifically, $s\ind {i_k}$ corresponds to $s\ind {j_k}$ for $k = 0,1,\dots,n'/2-1$. Note that \ttt Os are not parts of the bijection.

Next, we show the existence of a sequence of interleaving arcs.
\begin{theorem}\label{existence of interleaving arcs}
Let $s$ be an $\mathtt{AB}$-sequence of length $n$ such that $0$ is a valley. Then, there is a sequence of interleaving arcs corresponding to $s$. More specifically, the following algorithm returns an interleaving sequence of arcs for $s$.
\begin{algorithm}[H]
\caption{Constructing the Interleaving Arcs}
\begin{algorithmic}[1]
\State Let $a_0 < a_1 < \cdots < a_{n/2-1}$ be the indices of $\mathtt{A}$ in $s$
\State Let $b_0 < b_1 < \cdots < b_{n/2-1}$ be the indices of $\mathtt{B}$ in $s$
\State \Return $(r_n(a_k, b_k))_{k=0}^{n/2-1}$
\end{algorithmic}
\end{algorithm}
\end{theorem}
\begin{proof}
Because $0$ is a valley, any prefix substring contains at least as many \ttt As as \ttt Bs. That is, the first $b_k+1$ letters of $s$ contain exactly $k+1$ copies of \ttt B, so the first $b_k$ letters of $s$ contain at least $k+1$ copies of \ttt A, implying that $a_k < b_k$ for all $k$. Thus, the first condition for interleaving arcs is satisfied. The second and third conditions for interleaving arcs follow directly from the algorithm. \qed
\end{proof}

Finally, we show that if there is a sequence of interleaving arcs for $s$ such that all arcs have equal size, then $s$ is shift-compatible.
\begin{theorem}\label{AB shift-compatible}
    Let $d$ be a positive integer and let $s$ be an $\ttt{ABO}$-sequence where there exist interleaving arcs $\cal R$ such that each arc has size equal to $d$. Then, $s$ is shift-compatible with offset $d$.
\end{theorem}
\begin{proof}
    Since $s$ does not contain \ttt X, the shift-compatible condition becomes 
    \begin{align*}
        \text{If } s\ind i = \mtt A\ 
        &\text{then}\ s\ind{i+d \bmod n} = \mtt B\tag{*}\\
        \text{If } s\ind i = \mtt B\ 
        \text{or}\ s\ind{i} = \mtt O \ &\text{then}\ s\ind{i+d \bmod n} = \mtt A \ \text{or}\ s\ind{i+d \bmod n} = \mtt O \tag{**}
    \end{align*}

    Since each arc $r_n(i_k, j_k)$ has size $d$, we have $j_k = i_k + d$ for all $k$. By condition (iii) of interleaving arcs, $s\ind{i_k} = \mtt{A}$ and $s\ind{j_k} = \mtt{B}$, so $s\ind{i} = \mtt{A}$ if and only if $s\ind{i+d \bmod n} = \mtt{B}$. The condition $(*)$ is satisfied directly. The contrapositive of condition $(**)$ is that if $s \ind {i+d \bmod n} = \mtt B$, then $s\ind i = \mtt A$, which is true too. \qed
\end{proof}

\subsection{Constructing Insertion of \ttt Os}
In this section, we describe an algorithm for inserting $\mtt{O}$s into an \ttt{AB}-sequence, given an interleaving sequence of arcs, so that all arcs have equal size. Since a similar construction applies when $s$ is an \ttt{ABX}- or \ttt{ABXO}-sequence with additional restrictions, we introduce the following definition.
\begin{definition}[Inverted Indices]
    Let $s$ be an \ttt{ABXO}-sequence of length $n$. We call an index $t \in \mathbb{Z}_n$ \emph{inverted} if $s\ind{t} = \mtt{B}$ and $s\ind{t+1 \bmod n} = \mtt{A}$.
\end{definition}
\begin{theorem}\label{Inserting O}
    Let $s$ be an \ttt{AB}-sequence of length $n$ such that $0$ is a valley. Let $\cal R = (r_n(i_k,j_k))_{k=0}^{n/2-1}$ be a sequence of interleaving arcs for $s$. Let $(c_k)_{k=0}^{n/2-1}$ be a sequence of integers. Then, there is a sequence of non-negative integers $b_0,b_1,\dots,b_{n-1}$ and a common non-negative integer $d$ such that \[c_k + \sum_{l = i_k}^{j_k-1} b_{l} = d\] for all $k = 0,1,2,\dots, n/2-1$, and for all inverted indices $t$, we have $b_t = 0$.
\end{theorem}
\begin{remark}
    In the \ttt{AB}-sequence case, taking $c_k = j_k - i_k$ (the size of arc $r_n(i_k, j_k)$), the theorem implies that we can insert $b_i$ copies of $\mtt{O}$ between $s\ind{i}$ and $s\ind{i+1}$ for each $i = 0, 1, \dots, n-1$, so that all arcs in $\mathcal{R}$ have equal length $d$.
\end{remark}

\begin{proof}

\noindent\textbf{Step 1: Find a Particular Solution}

    This step finds a solution that allows $b_i$ to be negative.
    
    Fix $d = 0$. Since $j_0 < j_1 < \cdots < j_{n/2-1}$, each equation for $r_n(i_k, j_k)$ introduces a new free variable $b_{j_k-1}$ not appearing in any previous equation. We can therefore solve inductively to obtain an integer solution $b' = (b'_0, b'_1, \dots, b'_{n-1})$ satisfying \[
        c_k + \sum_{l=i_k}^{j_k-1} b'_l = 0
    \]
    for all $k = 0, 1, \dots, n/2-1$.
    
    Notably, only $b'_{j_k-1}$ is nonzero for each $k$. If $b'_t \ne 0$ where $t$ is an inverted index, we have that $t+1 = j_k$ for some $k$. Thus, $s\ind{t+1} = \mtt B$. However, since $t$ is an inverted index, $s\ind{t+1} = \mtt{A}$, which is a contradiction. Therefore, $b'_t = 0$ for all inverted indices $t$.
    
\noindent\textbf{Step 2: Find a Homogeneous Solution with Non-negative Coefficients}

    Consider the following system of equations \[
    \sum_{l=i_k}^{j_k-1} b_l\ind h = \varepsilon_h
    \]
    for all $k = 0,1,\dots, n/2-1$, where $\varepsilon_h \in \{0,1\}$ is a common value for all equations.
    
    For each non-inverted index $h$, we find a solution to that system of equations $b\ind h = (b_0\ind h,b_1\ind h,\dots,b_{n-1}\ind h)$ such that $b_h\ind h \geq 1$ and $b_t\ind h = 0$ for all inverted indices $t$ using the following algorithm. If some arc contains $h$, the algorithm returns a solution with $\varepsilon_h = 1$; otherwise, it returns a solution with $\varepsilon_h = 0$.

\begin{algorithm}[H]
\caption{Finding Homogeneous Solution for Column $h$}
\begin{algorithmic}[1]
\State $b_l\ind{h} \gets 0$ for all $l$; \quad $b_h\ind{h} \gets 1$
\State Find $k'$ such that $h \in r_n(i_{k'}, j_{k'})$ (if multiple exist, choose any one; if none exists, skip the loop and return $b\ind{h}$)
\For{$k \neq k'$ in order $k'-1, k'-2, \dots, 0, k'+1, k'+2, \dots, n/2-1$}
    \If{$\sum_{l=i_k}^{j_k-1} b_l\ind{h} = 0$}
        \State $b_{i_k}\ind{h} \gets 1$ \textbf{if} $k < k'$, \quad $b_{j_k-1}\ind{h} \gets 1$ \textbf{if} $k > k'$
    \EndIf
\EndFor
\State \Return $b\ind{h}$
\end{algorithmic}
\end{algorithm}
If no arc contains $h$, the returned vector has $b_h\ind h = 1$ as its only nonzero entry, so every arc sum is $0$ and the system is satisfied with $\varepsilon_h = 0$. Otherwise, the algorithm is correct by the following observation. At each iteration $k$, the equation for arc $r_n(i_k, j_k)$ 
is satisfied (i.e., $\sum_{l=i_k}^{j_k-1} b_l\ind{h} = 1$). Moreover, it remains satisfied in all subsequent iterations. This is because later iterations $k''$ only modify $b_{i_{k''}}\ind{h}$ or $b_{j_{k''}-1}\ind{h}$, which lie outside $r_n(i_k, j_k)$. Hence the sum over 
$r_n(i_k, j_k)$ is not affected.

Notably, only $b_{i_k}\ind{h}$ and $b_{j_k-1}\ind{h}$ can be nonzero. By the same argument as in Step~1, $b_t\ind{h} = 0$ for all inverted indices $t$.

\noindent\textbf{Step 3: Find a Positive Integer Solution}

    Let the system of $n/2$ linear equations $c_k + \sum_{l=i_k}^{j_k-1} b_l = d$ be written as $Ab + c = d\mathbf{1}$, where $A \in \{0,1\}^{n/2 \times n}$ is the interval matrix of the arcs (its $k$-th row indicates $r_n(i_k,j_k)$), $b = (b_i)_{i=0}^{n-1} \in \mathbb{Z}^n$, $c = (c_k)_{k=0}^{n/2-1} \in \mathbb{Z}^{n/2}$, $d \in \mathbb{Z}$, and $\mathbf{1} \in \mathbb{Z}^{n/2}$ is the all-ones vector.

    The first step shows that there exists $b' \in \bb Z^n$ such that $Ab' +c = 0$ and $b'_t = 0$ for all inverted indices $t$. The second step shows that for each non-inverted index $h$ there exists $b\ind h \in \bb N_0^n$ such that $Ab\ind h = \varepsilon_h{\bf 1}$ with $\varepsilon_h \in \{0,1\}$, $b_h \ind h \geq 1$, and $b\ind h_t = 0$ for all inverted indices $t$.
    
    Let $d_m = \max_i|b'_i|$. Consider $b^* = b' + d_m (\sum_{h~\text{not inverted}} b \ind h)$. If $i$ is not an inverted index, then \[b^*_i = b'_i + d_m \sum_{h~\text{not inverted}} b\ind h_i \geq b_i' + d_m \geq b_i' + |b'_i| \geq 0.\] If $i$ is an inverted index, then $b^*_i = 0$. Thus, $b^*$ has non-negative elements and $Ab^* + c = 0 + d_m n_t {\bf 1}  = d_m n_t{\bf 1}$ where $n_t$ is the number of non-inverted indices $h$ with $\varepsilon_h = 1$. Setting $d^* = d_m n_t$, the pair $(b^*, d^*)$ solves the system of equations. Moreover, since $b'_t$ and $b\ind h_t$ are always zero for all inverted indices $t$, we have that $b^*_t$ is also zero.

    As a result, $b^*$ satisfies the theorem requirement. \qed
\end{proof}

Finally, we finish the construction for equalization of $s$.
\begin{theorem}
    Let $s$ be an \ttt{AB}-sequence. Then, $s$ is \ttt O-shift-equalizable.
\end{theorem}
\begin{proof}
    From Lemma \ref{Existence of Valley}, we can construct $s_v$, which is a cyclic shift of $s$ such that $0$ is a valley. From Theorem \ref{existence of interleaving arcs}, $s_v$ has a sequence of interleaving arcs. By choosing $c_k = j_k-i_k$, Theorem~\ref{Inserting O} shows that we can obtain $s'$ by inserting \ttt Os into $s_v$ so that all arcs have equal size $d$. Finally, by Theorem \ref{AB shift-compatible}, $s'$ is shift-compatible.

    Since $s_v$ is a cyclic shift of $s$ and $s'$ is obtained from $s_v$ by inserting $\mtt{O}$s, we conclude that $s$ is $\mtt{O}$-shift-equalizable. \qed
\end{proof}
\begin{remark}\label{bound}
    After the insertion, the length of $s'$ can be bounded by $O(n^32^{n/2})$. More specifically, \[
    |b_{j_k-1}'| \leq |c_k| + \sum_{l=i_k}^{j_k-2} |b_l'| \leq |c_k| +\sum_{l=0}^{k-1} |b'_{j_l-1}|
    \]
    By induction, we have $|b'_{j_k-1}| \leq |c_k| + \sum_{l=0}^{k-1}2^{l}|c_{k-1-l}|$. Since $|c_k| \leq n$ for all $k$, we have $|b'_{j_k-1}| \leq n2^{k}$. Thus $d_m \leq n2^{n/2}$, so the offset is $d^* = d_m n_t = O(n^2 2^{n/2})$. Since each gap holds at most $d^*$ copies of \ttt O, the output length is $O(n \cdot d^*) = O(n^3 2^{n/2})$.
\end{remark}
\begin{remark}\label{sharper-bound}
We present a stronger bound for this problem, with a less constructive proof: using the matrix setup, there is a solution $b, d$ with offset $d \leq n \sum_{k} |c_k| = O(n^3)$. Since each gap holds at most $d$ copies of \ttt O, the output length is $O(n \cdot d) = O(n^4)$. The proof is given in Appendix~\ref{stronger_bound_proof}.

\end{remark}
\begin{example}
We run the complete construction on $s = \mtt{ABABBABA}$.

\smallskip
\noindent\textbf{Step 1 (shift to a valley).} Since $5$ is a valley index of $s$, we shift it so that index $0$ is a valley:
\[
s_v = \mtt{ABAABABB}.
\]

\smallskip
\noindent\textbf{Step 2 (interleaving arcs).} Applying Theorem~\ref{existence of interleaving arcs} to $s_v$, we get a collection of interleaving arcs \[
\cal R = \{r_8(0,1),r_8(2,4),r_8(3,6),r_8(5,7)\}
\]
\begin{figure}[H]
\centering
\begin{tikzpicture}[
  letter/.style={draw, circle, minimum size=8mm, inner sep=0pt, font=\ttfamily},
  arc/.style={-{Stealth[length=2mm]}, thick}]
  \foreach \i/\c in {0/A,1/B,2/A,3/A,4/B,5/A,6/B,7/B}
    \node[letter] (n\i) at (\i*1.1,0) {\c};
  \foreach \i in {0,...,7}
    \node[font=\scriptsize, text=gray] at (\i*1.1,-0.5) {\i};
  \draw[arc] (n0) to[bend left=70] (n1);
  \draw[arc] (n2) to[bend left=42] (n4);
  \draw[arc] (n3) to[bend left=58] (n6);
  \draw[arc] (n5) to[bend left=42] (n7);
\end{tikzpicture}
\caption{Interleaving arcs for $s_v = \mtt{ABAABABB}$: the $k$-th $\mtt A$ is
paired with the $k$-th $\mtt B$, giving $(0,1)$, $(2,4)$, $(3,6)$, $(5,7)$.
The $\mtt A$ endpoints increase from left to right, and so do the $\mtt B$
endpoints, so the arcs interleave.}
\label{fig:interleaving-arcs}
\end{figure}
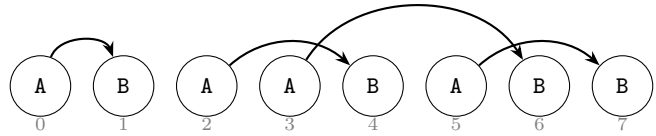
\smallskip
\noindent\textbf{Step 3 (insert $\mtt O$s to equalize).} 
Using Theorem~\ref{Inserting O}, we can insert \ttt O so that each arc has equal length, as shown in the following figure. That is, we can insert \ttt Os into $s_v$ to obtain $s'= \mtt{AOOBAOABABOB}$, which is shift-compatible. This implies that $s$ is \ttt O-shift-equalizable.
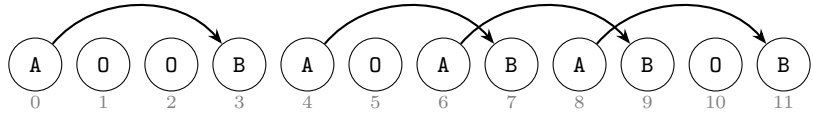
\begin{figure}[H]
\centering
\begin{tikzpicture}[
  letter/.style={draw, circle, minimum size=7mm, inner sep=0pt, font=\ttfamily\small},
  arc/.style={-{Stealth[length=2mm]}, thick}]
  \foreach \i/\c in {0/A,1/O,2/O,3/B,4/A,5/O,6/A,7/B,8/A,9/B,10/O,11/B}
    \node[letter] (m\i) at (\i*0.9,0) {\c};
  \foreach \i in {0,...,11}
    \node[font=\scriptsize, text=gray] at (\i*0.9,-0.5) {\i};
  \draw[arc] (m0) to[bend left=50] (m3);
  \draw[arc] (m4) to[bend left=50] (m7);
  \draw[arc] (m6) to[bend left=50] (m9);
  \draw[arc] (m8) to[bend left=50] (m11);
\end{tikzpicture}
\caption{After inserting $\mtt O$s into $s_v$ (two after index $0$, one after
index $2$, one after index $6$), we obtain $s' = \mtt{AOOBAOABABOB}$. Every arc
now has length $3$, so $s'$ is shift-compatible with offset $3$.}
\label{fig:step3-arcs}
\end{figure}
\end{example}

\section{Construction for \ttt{ABX}-Sequences}
In this section, we provide the construction of insertion for \ttt{ABX}-sequences. Intuitively, we replace each \ttt X in $s$ with \ttt{BA} and proceed as if $s$ were an \ttt{AB}-sequence. We will prove why this works.
\begin{theorem}\label{sufficient shift compatibility condition}
    Let $d$ be a positive integer and let $s$ be an \ttt{ABXO}-sequence of length $n$ such that if $s\ind{i} \in \{\mtt{A}, \mtt{X}\}$, then $s\ind{i+d \bmod n} \in \{\mtt{X}, \mtt{B}\}$. Then, $s$ is shift-compatible with offset $d$.
\end{theorem}

This theorem shows that the second condition for shift-compatibility (if $s\ind{i} \in \{\mtt{B}, \mtt{O}\}$, then $s\ind{i+d \bmod n} \in \{\mtt{A}, \mtt{O}\}$) follows automatically from the first condition.

\begin{proof}
    Consider the directed graph $G = (V,E)$ where $V = \{ 0,1,\dots,n-1\}$ and $E = \{(i, i+d \pmod n)\mid s\ind i\in\{\mtt A,\mtt X\}\}$. Consequently, if an index has non-zero in-degree, it must be either \ttt X or \ttt B.
    
    Notice that this is a subgraph of $G' = (V,E')$ where $E' = \{(i,i+d \bmod n) \mid i \in \bb Z_n\}$, which contains only cycles. Therefore, each connected component of $G$ must be a chain (a component with at least one edge that is not a cycle), a cycle, or an isolated vertex.
    \begin{itemize}
        \item If it is a cycle, every vertex must have non-zero in-degree and out-degree. The only possibility is that all vertices are indices of \ttt X. That is because other letters either have zero in-degree, zero out-degree, or both.
        \item If it is an isolated vertex, it has zero in-degree and zero out-degree, so its letter must be \ttt B or \ttt O. We will see below that it must in fact be \ttt O.
        \item If that connected component is a chain, the last vertex must be an index of \ttt B because of its zero out-degree but non-zero in-degree. Moreover, \ttt A must be the start of a chain because of its zero in-degree but non-zero out-degree. Hence, $|s|_{\mtt A} \leq \text{the number of chains} \leq |s|_{\mtt B}$.  Because $|s|_{\mtt A} = |s|_{\mtt B}$ from the definition of \ttt{ABXO}-sequence, the equality holds, so every chain starts with \ttt A and ends with \ttt B. In particular, every \ttt B is the end of a chain, so no \ttt B is an isolated vertex.
    \end{itemize} 

    This implies that, if $s\ind{i+d \bmod n} \in \{\mtt X, \mtt B\}$, we have that $s\ind i \in \{\mtt A, \mtt X\}$. The second shift-compatibility condition follows by contrapositive.
    \qed
\end{proof}

\begin{theorem} \label{ABX equalizable}
    Let $s$ be an \ttt{ABX}-sequence. Then, $s$ is \ttt O-shift-equalizable.
\end{theorem}
\begin{proof}
    Suppose that $|s|_{\mtt X} = m$ and the length of $s$ is $n$. Let $s^{e}$ (the \emph{expansion} of $s$) be the \ttt{AB}-sequence constructed by replacing each \ttt{X} in $s$ with \ttt{BA}. Thus, the length of $s^{e}$ is $n^{e} = n + m$.
    
    Similar to the construction for \ttt{AB}-sequences, Lemma \ref{Existence of Valley} guarantees that there is a cyclic shift $s^{e}_v$ of $s^{e}$ such that $0$ is a valley of $s^{e}_v$. Let $t_0<t_1<\dots<t_{m-1}$ be the inverted indices of $s^{e}_v$ caused by expanding \ttt X to \ttt{BA} in the earlier step. Call those $m$ indices \emph{induced inverted indices}.
    
    From Theorem \ref{existence of interleaving arcs}, $s^{e}_v$ has a sequence of interleaving arcs $\cal R^{e} = (r_{n^{e}}(i^{e}_k,j^{e}_k))_{k=0}^{n^{e}/2-1}$. Let $v_k$ be the number of induced inverted indices in $r_{n^e}(i^e_k,j^e_k)$. By choosing $c_k = j^e_k-i^e_k - v_k$, there is a way to obtain ${s^e_v}'$ by inserting \ttt Os into $s^e_v$ so that the length of each arc minus the number of induced inverted indices contained in that arc is the same for all arcs and no \ttt Os are inserted after induced inverted indices.

    After that, we obtain $s_v'$ by reducing the pairs \ttt{BA} in ${s_v^e}'$ that we expanded earlier to be \ttt X. This reduction is possible because we did not insert any \ttt O between those pairs. Because of this reduction, the sequence of interleaving arcs $\cal R^e$ induces a collection of arcs $\cal R$ in $s_v'$. 

    Notice that $\cal R = (r_n(i_k,j_k))_{k=0}^{n^e/2-1}$ is generally similar to a sequence of interleaving arcs, except that at most one arc may wrap around index $0$. More specifically, $\cal R$ has the following property:
    \begin{itemize}
        \item Each letter \ttt A is the start of an arc in $\cal R$, each letter \ttt B is the end of an arc in $\cal R$, and each letter $\ttt X$ is both the start of one arc and the end of another arc. 
        \item Each arc in $\cal R$ has cyclic length exactly $d$, that is, $j_k \equiv i_k + d \pmod{n}$, where $d$ is the common value from Theorem~\ref{Inserting O}.
        \item At most one arc wraps around index $0$. This happens only when an induced pair itself wraps around the word, that is, its \ttt B is at the last index and its \ttt A is at index $0$. In that case, the arc ending at that \ttt B now ends at the merged \ttt X at index $0$. All other arcs satisfy $i_k < j_k$.
    \end{itemize}
    The first condition is true because \ttt X functions as both \ttt A and \ttt B. The second condition is true because of the choice of $c_k$: reducing an induced pair contained in an arc shortens that arc by exactly one, and no \ttt O was inserted between the letters of a pair. The wrapping pair is never contained in any arc, since every arc's index set ends at $j_k - 1 \le n^e - 2$; reducing it moves the last position onto index $0$ without changing the cyclic length of any arc.
    
    From Theorem \ref{sufficient shift compatibility condition}, the existence of $\cal R$ implies that the word $s_v'$ is shift-compatible.

    So far, given $s$, we expand it to $s^e$ and cyclically shift it to $s^e_v$. We insert \ttt Os into $s^e_v$ to obtain ${s^e_v}'$ without separating any \ttt{BA} pairs, and we reduce it to $s_v'$. The operations ensure that we can insert \ttt Os into $s$ so that it is cyclically equivalent to $s_v'$. Therefore, $s$ is \ttt O-shift-equalizable.
\end{proof}
\begin{remark}
    Note that the proof also works when the valley is an index of \ttt A created by the expansion of \ttt X.
\end{remark}

\begin{example}
We run the construction on the \ttt{ABX}-sequence $s = \mtt{AXBABX}$.

\smallskip
\noindent\textbf{Step 1 (expand each $\mtt X$ to $\mtt{BA}$).} Replacing both $\mtt X$s gives
\[
s^e = \mtt{ABABABBA}.
\]
Shifting so that the first index is a valley, we obtain
\[
s^e_v = \mtt{AABABABB}.
\]

\smallskip
\noindent\textbf{Step 2 (find arcs and insert $\mtt O$).} The interleaving arcs of $s^e_v$ are
\[
\cal R = \{r_8(0,2),\, r_8(1,4),\, r_8(3,6),\, r_8(5,7)\},
\]
Inserting \ttt O gives
\[
{s^e_v}' = \mtt{AAOBABABOB}.
\]
Note that this sequence is not shift-compatible yet.
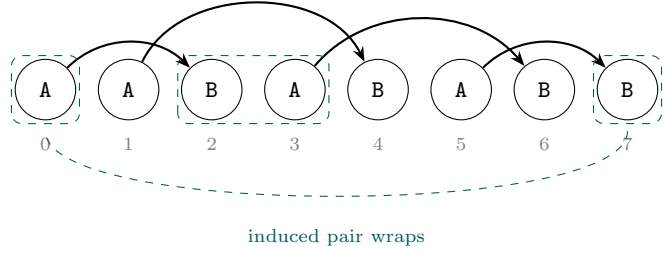
\begin{figure}[H]
\centering
\begin{tikzpicture}[
  letter/.style={draw, circle, minimum size=8mm, inner sep=0pt, font=\ttfamily},
  arc/.style={-{Stealth[length=2mm]}, thick},
  frame/.style={teal!70!black, dashed, rounded corners}]
  \foreach \i/\c in {0/A,1/A,2/B,3/A,4/B,5/A,6/B,7/B}
    \node[letter] (n\i) at (\i*1.1,0) {\c};
  \foreach \i in {0,...,7}
    \node[font=\scriptsize, text=gray] at (\i*1.1,-0.72) {\i};
  % arcs
  \draw[arc] (n0) to[bend left=45] (n2);
  \draw[arc] (n1) to[bend left=64] (n4);
  \draw[arc] (n3) to[bend left=50] (n6);
  \draw[arc] (n5) to[bend left=45] (n7);
  % middle induced pair (indices 2,3) -- sits inside arc (1,4)
  \draw[frame] (1.75,-0.45) rectangle (3.75,0.45);
  % wrapping induced pair (indices 7 and 0)
  \draw[frame] (-0.45,-0.45) rectangle (0.45,0.45);
  \draw[frame] (7.25,-0.45) rectangle (8.15,0.45);
  \draw[frame] (7.7,-0.55) .. controls (7.7,-1.7) and (0,-1.7) .. (0,-0.55);
  \node[font=\scriptsize, text=teal!70!black] at (3.85,-1.95) {induced pair wraps};
\end{tikzpicture}
\caption{Step 2 for $s^e_v = \mtt{AABABABB}$: arcs $(0,2),(1,4),(3,6),(5,7)$. Notice that this induced pair wraps around the end and the start of the word. This does not break our algorithm.}
\label{fig:axbabx-step2}
\end{figure}
\begin{figure}[H]
\centering
\begin{tikzpicture}[
  letter/.style={draw, circle, minimum size=8mm, inner sep=0pt, font=\ttfamily},
  arc/.style={-{Stealth[length=2mm]}, thick},
  frame/.style={teal!70!black, dashed, rounded corners}]
  \foreach \i/\c in {0/A,1/A,2/O,3/B,4/A,5/B,6/A,7/B,8/O,9/B}
    \node[letter] (n\i) at (\i*1.0,0) {\c};
  \foreach \i in {0,...,9}
    \node[font=\scriptsize, text=gray] at (\i*1.0,-0.72) {\i};
  % arcs (arc (1,5) is one longer)
  \draw[arc] (n0) to[bend left=42] (n3);
  \draw[arc] (n1) to[bend left=64] (n5);
  \draw[arc] (n4) to[bend left=42] (n7);
  \draw[arc] (n6) to[bend left=42] (n9);
  % middle induced pair (indices 3,4) -- inside arc (1,5)
  \draw[frame] (2.55,-0.45) rectangle (4.45,0.45);
  % wrapping induced pair (indices 9 and 0)
  \draw[frame] (-0.45,-0.45) rectangle (0.45,0.45);
  \draw[frame] (8.55,-0.45) rectangle (9.45,0.45);
  \draw[frame] (9.0,-0.55) .. controls (9.0,-1.7) and (0,-1.7) .. (0,-0.55);
  \node[font=\scriptsize, text=teal!70!black] at (4.5,-1.95) {induced pair wraps};
\end{tikzpicture}
\caption{After inserting $\mtt O$ (at indices $2,8$): ${s^e_v}' = \mtt{AAOBABABOB}$,
with arc lengths $3,4,3,3$. The long arc $(1,5)$ still holds the middle induced
pair, so the word is not yet shift-compatible.}
\label{fig:axbabx-after-O}
\end{figure}
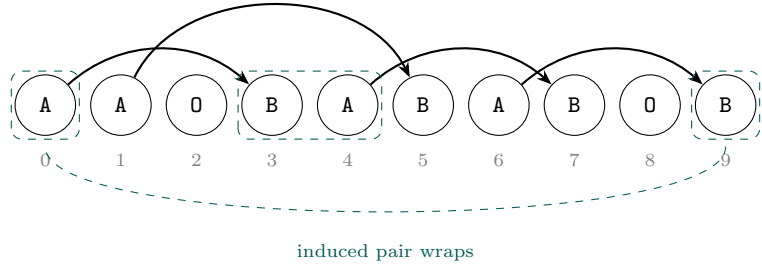
\smallskip
\noindent\textbf{Step 3 (collapse $\mtt{BA}$ back to $\mtt X$).} Collapsing both induced pairs to $\mtt X$ shrinks the long arc to match the rest and yields
\[
s' = \mtt{XAOXBABO},
\]
which is shift-compatible with offset $3$. Hence $s$ is \ttt O-shift-equalizable.
\end{example}

\section{Construction for \ttt{ABXO}-Sequences}
Finally, we provide the construction for equalizing any \ttt{ABXO}-sequence.
\begin{theorem}
    Let $s$ be an \ttt{ABXO}-sequence. Then, $s$ is \ttt O-shift-equalizable.
\end{theorem}
\begin{proof}
    If $s$ only contains the letter \ttt O, the theorem is trivial. Thus, assume that $s$ has at least one non-\ttt O letter. Because a cyclic shift does not affect the shift-equalizability, without loss of generality, $s$ starts with a non-\ttt O letter.
    
    Let $\ell_{\mtt O}$ be the length of the longest run of $\mtt O$ in $s$. Let $s_r$ be the \ttt{ABX}-sequence obtained by removing all \ttt O from $s$. By Theorem \ref{ABX equalizable}, we can insert \ttt O into $s_r$ to obtain $s_r'$ so that $s_r'$ is shift-compatible with offset $d$.

    Then, we obtain $s'$ by inserting $\ell_{\mtt O}$ copies of \ttt O after each character in $s_r'$ (including the end). That is, ${s_r'}\ind i = {s'} \ind{i \times (\ell_{\mtt O} + 1)}$. This construction implies that $s'$ is shift-compatible with offset $(\ell_{\mtt O}+1)\times d$: a step of size $d$ in $s_r'$ corresponds to a step of size $(\ell_{\mtt{O}}+1) \times d$ in $s'$.
    
    By the maximality of $\ell_{\mtt{O}}$, the original \ttt{O}s in $s$ can be accommodated within the inserted copies, so $s'$ is obtained from $s$ by inserting \ttt{O}s. Therefore, $s$ is \ttt O-shift-equalizable, as desired.
\end{proof}
 \begin{example}
We run the construction on the \ttt{ABXO}-sequence $s = \mtt{AOXBOAOOBX}$.

\smallskip
\noindent\textbf{Step 1 (remove every $\mtt O$).} Deleting all $\mtt O$ leaves
the \ttt{ABX}-sequence
\[
s_r = \mtt{AXBABX}.
\]

\smallskip
\noindent\textbf{Step 2 (equalize $s_r$).} By the previous example, inserting
$\mtt O$ into $s_r$ gives
\[
s_r' = \mtt{XAOXBABO},
\]
shift-compatible with offset $d = 3$.

\smallskip
\noindent\textbf{Step 3 (pad with $\mtt O$).} The longest $\mtt O$-run in $s$ is
$\mtt{OO}$, so $\ell_{\mtt O} = 2$. Insert two $\mtt O$s after each character of
$s_r'$:
\[
s' = \mtt{XOO\,AOO\,OOO\,XOO\,BOO\,AOO\,BOO\,OOO}
   = \mtt{XOOAOOOOOXOOBOOAOOBOOOOO},
\]
shift-compatible with offset $(\ell_{\mtt O}+1)\,d = 3 \times 3 = 9$. The original
$\mtt O$s of $s$ fit inside the inserted copies, so $s$ is \ttt O-shift-equalizable.
\end{example}

\section{Consequences}\label{stronger_bound}
The Main Theorem (Theorem~\ref{main_theorem}) can be restated as follows: any two binary words that are permutations of each other can be made cyclically equivalent by a single $0$-simultaneous insertion. The same statement holds for $1$-insertion, by complementing both words.

The construction is explicit, so it also bounds the length of the resulting words. By Remark~\ref{bound}, equalizing an \ttt{AB}-sequence yields output length $O(n^3 2^{n/2})$, which remains exponential through the \ttt{ABX} and \ttt{ABXO} stages. The sharper analysis of Appendix~\ref{stronger_bound_proof} removes the exponential factor: an \ttt{AB}-sequence has offset $d = O(n^3)$ and hence output length $O(n \cdot d) = O(n^4)$; expanding \ttt X preserves this degree, and the \ttt O-runs of length $O(n)$ then give output length $O(n^5)$ for an \ttt{ABXO}-sequence. We do not claim this bound is tight.

\section{Application to Card-Based Cryptography}
We return to the card model that motivated the theorem. In the standard two-colored setting~\cite{denboer,mizuki-shizuya}, a $\clubsuit$ encodes $0$ and a $\heartsuit$ encodes $1$, so a face-down sequence is a binary word, and a random cut applies a uniformly random cyclic shift. Two sequences are thus indistinguishable after a random cut if and only if they are cyclic shifts of each other.

Our result is a tool for protocol designers. Our theorem shows which card configurations a designer can make indistinguishable, and at what cost. Let $C_1$ and $C_2$ be two sequences of equal length.

\begin{itemize}
  \item If they contain the same number of $\heartsuit$ cards, they can always be made indistinguishable under a random cut, using a single type of helping card ($\clubsuit$).
  \item If they contain different numbers of $\heartsuit$ cards, no insertion can achieve this, regardless of how many cards or types are used.
\end{itemize}

The construction may insert many $\clubsuit$ cards, up to $O(n^5)$ for sequences of length $n$. Reducing this number
remains open.

\section{Open Problems}
We close with two questions left open by this work.

\subsection{Sharpness of the length bound}

Our construction equalizes two binary words of length $n$ while increasing their length to $O(n^5)$ (Appendix~\ref{stronger_bound_proof}). We do not know whether this polynomial bound is sharp.

\begin{problem}\label{prob:bound}
What is the exact, non-asymptotic length needed to equalize two binary words of length $n$? In particular, is there a family of words that requires $\Omega(n^3)$ inserted symbols, and can the $O(n^5)$ upper bound be lowered to meet such a lower bound?
\end{problem}

\subsection{Discarding a letter over larger alphabets}

Our binary result shows that one letter is never needed: the letter $1$ can be discarded. We ask whether one letter can always be discarded over larger alphabets.

The discarded letter cannot be fixed in advance. Over $\Sigma = \{0,1,2\}$, the words $012$ and $021$ are permutations of each other, but they are not $\{1,2\}$-equalizable: each word contains a single $0$, so the two $0$s must align, and the other letters can thus never be made to match.

\begin{problem}\label{prob:discard}
Let $\Sigma$ be an alphabet of size at least $3$. Let $u$ and $v$ be words over $\Sigma$ that are permutations of each other. Must there exist a letter $c \in \Sigma$ such that $u$ and $v$ are $(\Sigma \setminus \{c\})$-equalizable?

\end{problem}

\section*{Acknowledgements}
The author is grateful to Suthee Ruangwises for introducing him to this line of research and for his advice on publication venues.

\appendix
\section{Polynomial Length Bound}
\label{stronger_bound_proof}
We now give a sharper bound on the amount of insertion. Remark~\ref{bound} bounds
the output length of an \ttt{AB}-sequence by $O(n^3 2^{n/2})$, an exponential
quantity. Using linear algebra, we replace the exponential factor by a polynomial
one. The construction is exactly that of Theorem~\ref{Inserting O}; only the
particular solution of Step~1 is chosen differently.

\begin{theorem}
Let $s$ be an \ttt{AB}-sequence of length $n$ such that $0$ is a valley, let
$\cal R = (r_n(i_k, j_k))_{k=0}^{m-1}$ be a sequence of interleaving arcs for $s$
with $m = n/2$, and let $A \in \{0,1\}^{m \times n}$ be the matrix whose $k$-th
row indicates the arc $r_n(i_k, j_k)$, that is, $A_{k,\ell} = 1$ if and only if
$i_k \le \ell < j_k$ (the interval matrix of Theorem~\ref{Inserting O}). For any
$c = (c_k)_{k=0}^{m-1} \in \bb Z^{m}$, there exist $b \in \bb N_0^{\,n}$ and
$d \in \bb N_0$ such that
\[
A b + c = d\,\mathbf{1}_{m}, \qquad\text{and}\qquad
b_t = 0 \text{ for every inverted index } t,
\]
with $d \le n \sum_{k=0}^{m-1} |c_k|$.
\end{theorem}

\begin{proof}
Let $T \subseteq \bb Z_n$ be the set of inverted indices, and let $A_R$ be the
$m \times (n - |T|)$ matrix obtained by deleting from $A$ the columns indexed by
$T$. Then $A_R$ is again an interval matrix, since deleting columns leaves each
row a block of consecutive $1$s.

We show that $A_R$ still has full row rank. By Step~1 of the proof of
Theorem~\ref{Inserting O}, column $j_k - 1$ is the unique pivot of row $k$, and
the index $j_k - 1$ is not inverted: if it were, then $s\ind{j_k} = \mtt{A}$,
contradicting $s\ind{j_k} = \mtt{B}$. Hence no pivot column is deleted, so each
row of $A_R$ still has its own pivot, and the rows of $A_R$ remain linearly
independent.

Append to $A_R$ the standard basis row $e_\ell$ for each non-pivot column $\ell$,
forming a square invertible matrix $A_R' \in \{0,1\}^{(n - |T|) \times (n - |T|)}$.
Since $A_R'$ is still an interval matrix, it is totally
unimodular~\cite{schrijver1986}; consequently every entry of ${A_R'}^{-1}$ lies in
$\{0, +1, -1\}$.

Let $c_p$ be $c$ padded with zeros to length $n - |T|$, and put
$b'_p = -{A_R'}^{-1} c_p \in \bb Z^{n-|T|}$, so that $A_R' b'_p + c_p = 0$.
Because each entry of ${A_R'}^{-1}$ has magnitude at most $1$,
\[
|b'_{p,i}| = \Big| \sum_{j} ({A_R'}^{-1})_{ij}\, c_{p,j} \Big|
\le \sum_{j} |c_{p,j}| = \sum_{k=0}^{m-1} |c_k|.
\]
Discarding the appended rows of $A_R'$ and the padded zeros of $c_p$ leaves
$A_R b'_p + c = 0$. Inserting a zero into $b'_p$ at each inverted index gives
$b' \in \bb Z^n$ with $A b' + c = 0$ and $b'_t = 0$ for all inverted indices $t$.

Exactly as in Steps~2 and 3 of the proof of Theorem~\ref{Inserting O}, adding a
non-negative homogeneous solution makes all entries non-negative without altering
the values on $T$, producing $b^* \in \bb N_0^{\,n}$ with $b^*_t = 0$ for all
$t \in T$ and $A b^* + c = d_m n_t\,\mathbf{1}_m$, where $d_m = \max_i |b'_i|$ and
$n_t \le n$ is as in Step~3 of that proof. Setting $d^* = d_m n_t$, the
pair $(b^*, d^*)$ solves the system, and
\[
d^* = d_m n_t \le n \max_i |b'_i| \le n \sum_{k=0}^{m-1} |c_k|.
\]
\qed
\end{proof}

Taking $c_k = j_k - i_k$ as in Remark~\ref{bound} gives
$\sum_k |c_k| = \sum_k (j_k - i_k) = O(n^2)$, so the offset satisfies
$d = O(n^3)$, replacing the exponential offset of Remark~\ref{bound} by a
polynomial one. Since each gap holds at most $d$ copies of \ttt O, the output
length of an \ttt{AB}-sequence is $O(n \cdot d) = O(n^4)$; the \ttt O-runs of the
\ttt{ABXO} construction then multiply this by $O(n)$, giving output length
$O(n^5)$.

\bibliographystyle{plain}
\bibliography{references}

\end{document}